\documentclass[sigconf]{acmart}
\usepackage[activeacute,english]{babel} %Languages
\usepackage[utf8]{inputenc}
\usepackage{verbatim} %allows one to comment large chunks
\usepackage{graphicx} %For inserting images
\usepackage{enumerate,mdwlist}  %For list and personalization of lists
\usepackage{multirow} %For the tables
\usepackage{amsmath,amsfonts,amsthm} %AMSpackages
\usepackage{bm} %Bold math symbols
\usepackage{empheq} %For braces in systems of equations
\usepackage{bbm,dsfont} %Blackboard symbols
\usepackage{mathrsfs} % script letters
\numberwithin{equation}{section} %Numbering of equations
\usepackage{float} % for the position of the figures
\usepackage{empheq} %for systems with braces

\usepackage[OT2,OT1]{fontenc}
\DeclareRobustCommand\cyr{%
  \renewcommand\rmdefault{wncyr}%
  \renewcommand\sfdefault{wncyss}%
  \renewcommand\encodingdefault{OT2}%
  \normalfont
  \selectfont}
\DeclareTextFontCommand{\textcyr}{\cyr}

\usepackage{subcaption}

\def\mcB{\mathcal{B}}

\def\bbR{\mathbb{R}}

\def\bbN{\mathbb{N}}

\def\bbC{\mathbb{C}}
\def\bbP{\mathbb{P}}

\def\fkx{\mathfrak{x}}
\def\fky{\mathfrak{y}}
\def\fku{\mathfrak{u}}

\def\fkt{\mathfrak{t}}
\def\fks{\mathfrak{s}}

\def\Oh{\mathcal{O}}

\DeclareMathOperator{\dist}{\mathrm{dist}}
\DeclareMathOperator{\distTV}{\mathrm{dist}_{\mathrm{TV}}}
\def\chev{\textrm{\cyr Ch}}
\def\chevI{\textrm{\cyr Ch}I}
\def\bcomp{\textrm{\cyr B}}
\def\cond{\texttt{C}}
\def\enumber{\mathrm{e}}

\DeclarePairedDelimiter\norm{\lVert}{\rVert}

\title[On the Error of Random Sampling: Uniformly Distributed Random Points on Parametric Curves]{On the Error of Random Sampling:\texorpdfstring{\\}{ }Uniformly Distributed Random Points on Parametric Curves}
\author{Apostolos Chalkis}
\affiliation{%
\institution{National \& Kapodistrian Univ. Athens\\ GeomScale org}%National and Kapodistrian University of Athens
\streetaddress{Panepistimiopolis}
\city{Athens}
\postcode{15784}
\country{Greece}}
\email{achalkis@di.uoa.gr}

\author{Christina Katsamaki}
\affiliation{
\institution{Inria Paris \& IMJ-PRG\\Sorbonne Universit\'e}
\streetaddress{4 place Jussieu}
\city{Paris}
\postcode{F-75005}
\country{France}}
\email{chistina.katsamaki@inria.fr}

\author{Josué Tonelli-Cueto}
\orcid{0000-0002-2904-1215}
\affiliation{
\institution{Inria Paris \& IMJ-PRG\\Sorbonne Universit\'e}
\streetaddress{4 place Jussieu}
\city{Paris}
\postcode{F-75005}
\country{France}}
\email{josue.tonelli.cueto@bizkaia.eu}

\begin{abstract}
Given a parametric polynomial curve $\gamma:[a,b]\rightarrow \mathbb{R}^n$, how can we sample a random point $\mathfrak{x}\in \mathrm{im}(\gamma)$ in such a way that it is distributed uniformly with respect to the arc-length? Unfortunately, we cannot sample exactly such a point—even assuming we can perform exact arithmetic operations. So we end up with the following question: how does the method we choose affect the quality of the approximate sample we obtain? In practice, there are many answers. However, in theory, there are still gaps in our understanding. In this paper, we address this question from the point of view of complexity theory, providing bounds in terms of the size of the desired error.
\end{abstract}

\date{}

\makeatletter
\def\th@plain{%
  \thm@notefont{}% same as heading font
  \slshape % body font
}
\def\th@definition{%
  \thm@notefont{}% same as heading font
  \normalfont % body font
}
\makeatother
\theoremstyle{plain}
\newtheorem{lem}{Lemma}[section]
\newtheorem{prop}[lem]{Proposition}
\newtheorem{theo}[lem]{Theorem}

\theoremstyle{definition}
\newtheorem{defi}[lem]{Definition}
\theoremstyle{remark}

\newtheorem{remark}[lem]{Remark}
\newcommand{\eproof}{\hfill\qed}

\usepackage[linesnumbered,ruled,vlined]{algorithm2e}
\SetKwInOut{given}{Given}
\SetKwComment{Comment}{/* }{ */}

\SetCommentSty{mycommfont}
\let\oldnl\nl% Store \nl in \oldnl
\newcommand{\nonl}{\renewcommand{\nl}{\let\nl\oldnl}}% Remove line number for one line

\makeatletter
\let\original@algocf@latexcaption\algocf@latexcaption
\long\def\algocf@latexcaption#1[#2]{%
  \@ifundefined{NR@gettitle}{%
    \def\@currentlabelname{#2}%
  }{%
    \NR@gettitle{#2}%
  }%
  \original@algocf@latexcaption{#1}[{#2}]%
}
\makeatother

\renewcommand\footnotetextcopyrightpermission[1]{} % removes footnote with conference information in first column
\setcopyright{none}
\copyrightyear{2022}
\acmYear{2022}
\setcopyright{acmlicensed}
\acmConference[]{}{}{}
\acmBooktitle{}%International Symposium on Symbolic and Algebraic Computation (ISSAC '22), July 4--7, 2020=2, Lille, France}
\acmPrice{15.00}
\acmDOI{}
\acmISBN{}


\begin{document}

\maketitle
\begin{abstract}
Given a parametric polynomial curve $\gamma:[a,b]\rightarrow \mathbb{R}^n$, how can we sample a random point $\mathfrak{x}\in \mathrm{im}(\gamma)$ in such a way that it is distributed uniformly with respect to the arc-length? Unfortunately, we cannot sample exactly such a point—even assuming we can perform exact arithmetic operations. So we end up with the following question: how does the method we choose affect the quality of the approximate sample we obtain? In practice, there are many answers. However, in theory, there are still gaps in our understanding. In this paper, we address this question from the point of view of complexity theory, providing bounds in terms of the size of the desired error.
\end{abstract}

%%% Local Variables:
%%% mode: latex
%%% TeX-master: "main"
%%% End:

\section{Introduction}

Given a parametric polynomial curve $\gamma: I:=[a,b]\rightarrow \bbR^n$, we are interested in generating a random point $\fkx\in \gamma(I)$ that is uniformly distributed with respect the arc-length. To do this, we only need to sample a random variable $\fkt\in I$ with density proportional to the speed of the curve $\norm{\gamma'}_2$. Even if we perform exact arithmetic operations using real numbers with infinite precision, this problem does not admit an exact solution---the integral $\int \norm{\gamma'}_2$ cannot be expressed in terms of elementary functions. The goal of this paper is to estimate how much the generated random sample differs from the one that we want and how does the desired error affect the complexity.

\subsection{Random Samples: Why do we care?}

At this point, to avoid possible confusion, we want to clarify that we are studying random sampling on a parametric curve. Unfortunately, in the literature, the term `sample' carries two meanings. In the context of computational statistics, this refers to generating a random point~\cite{samplingVempala}. In the context of parametric curves, this refers to generating a finite subset of points that captures the curve~\cite{paganiscott2018}. Of course, random sampling can be used for generating samples in the second sense, but it is not the best method since randomness might produce points too near to the ones already produced.

Now, being able to sample random points in an algebraic variety plays an important role in the application of topological data analysis (TDA) to algebraic geometry~\cite{breidingkalisniksturmfelsweinstein2018}---if we sample enough random points on an algebraic variety, we can determine its topology as shown in Figure~\ref{fig:curve_sampled}. However, in order to bridge the gap between the theoretical assumptions on the random samples (as those stated in~\cite{niyogismaleweinberger2008}) and the random samples that we actually generate, we have to understand how imperfect our generation of these random samples is. This paper is a first step towards filling this gap in our theoretical understanding in the simplest case: sampling random points on parametric polynomial curves.

\subsection{Errors in Random Sampling: State of the art}

Sampling from a density is an important and well-studied problem in computational statistics. There are many methods to sample according to a certain distribution: Acceptance-Rejection (AR) method~\cite{Leydold98}, Adaptive Rejection Sampling~\cite{Gilks92}, Slice Sampling~\cite{Neal03}, etc. Currently, the state-of-the-art samplers are the so-called
Markov Chain Monte Carlo (MCMC) algorithms. They have plenty of success stories, such as Hamiltonian Monte Carlo~\cite{NealMCMC}, Hit-and-Run~\cite{Smith96} and Metropolis-Hastings~\cite{chibgreenberg1995}. 

However, most of the aforementioned methods apply only for log-concave distributions~\cite{lovasz2006hit,Lovasz06,chen2019optimal,lee2018algorithmic,dwivedi2019log}. The algorithm in~\cite{pmlr-v99-mangoubi19a} can be used to sample from (multivariate) non-convex density functions, but it does not handle the case where the density is restricted to an interval (or a bounded set in general). Moreover, for the univariate setting, the existing MCMC algorithms~\cite{Johndrow18} either do not provide any error guarantees or could lead to arbitrarily high run-times. In particular, this means that the error-control for generating random samples on a parametric curve is an open problem.

\subsection{Analyzed Method and Contributions}

We analyze the method proposed by Olver and Townsed~\cite{Olver13}. This method is not limited to random sampling on parametric curves, but it holds for general density functions. They demonstrate empirically its efficiency, however, no theoretical analysis is given. We aim to do this in the particular case of sampling random points on parametric polynomial curves, uniformly with respect to the arc-length.

The underlying idea of Olver and Townsed~\cite{Olver13} is to use Chebyshev approximations to make inverse transform sampling tractable at the cost of an error in the produced random sample. The cumulative distribution function might not be expressible in terms of elementary functions---and this happens in the case of interest.
%since solving a polynomial is always better than solving As solving directly $\Phi(\fkx)=\fku$ for some uniformly distributed $\fku\in[0,1]$ and the  $\Phi$ of the considered distribution might not be possible, they use Chebyshev polynomials to substitute $\Phi$ by a suitable approximation $\tilde{\Phi}$ easier---and faster---to handle.

We provide the first error analysis for this method in the case where we are sampling random points on a parametric polynomial curve. %In other words, we study the method when we are sampling a random $\fkt\in I$ with density function is proportional to $\|\gamma'\|_2$ where $\gamma:I\rightarrow \bbR^n$ is a parametric polynomial curve.
For this error analysis, we work in the BSS model (see \S\ref{sec:BSS}) for convenience. This model is suitable for developing a complexity theory over the real numbers, since a BSS machine is like a Turing machine, but it can operate with real numbers and it performs arithmetic operations and comparisons at unit cost. The main result of this paper is the following one:

\begin{theo}
Let $\gamma :I \rightarrow \mathbb{R}^n$ be a polynomial parameterized curve of degree $d$. The algorithm \nameref{alg:newsampler} samples points from $\gamma$ uniformly with respect to the arc-length by performing
\[
\Oh(\ell^3(1+\log d\cond(\gamma))^3d^3\cond(\gamma)^3)
\]
arithmetic operations and with error $2^{-\ell}$ with respect to the total variation distance, where $\cond(\gamma)$, the condition number for sampling $\gamma$, is given Definition~\ref{defi:conditiongammasampling}.
\end{theo}

Our analysis is generic enough to be applicable to more general densities, but we focus on parametric curves since this is the case of interest in our research program. This paper is the first one studying reductions of randomness sources in the BSS model of computation. Even more, we provide an open-source {\tt Matlab} implementation (see \S\ref{sec:implementation}) together with an extended experimental analysis, which confirms the theoretical results.

\subsection{Related Problems}

Although we are analyzing the method of~\cite{Olver13}, we note that our problem ---generating random points in parametric polynomial curves--- is related to a lot of problems in the literature. On the one hand, this problem is related to obtaining arc-length parametrizations and generating deterministic uniform samples---points that are equidistributed with respect the arc-length---of parametric curves. For these problems, there is an extensive literature~\cite{paganiscott2018,figueiredo1995,floaterrasmussen2006,floaterrasmussenreif2007,gravesen1997,piegltiller1997NURBSbook,walterfournier1996} and an analysis of these methods might be possible following our strategy---we only have to control the following $L^1$-norm:
\[
\left\|\|\gamma'(\theta(s))\|_2\theta'(s)-\chi_{[0,L]}(s)\right\|_{1},
\]
where $\theta:[0,L]\rightarrow I$ is a reparametrization of $\gamma:I\rightarrow\bbR^n$. However, we feel that these methods might not generalize easily to higher dimensions, that we plan to deal with in the future. Nevertheless, we leave for future work a careful study of these methods.

\begin{figure}[!t]
	\centering
		\includegraphics[width=0.45\textwidth]{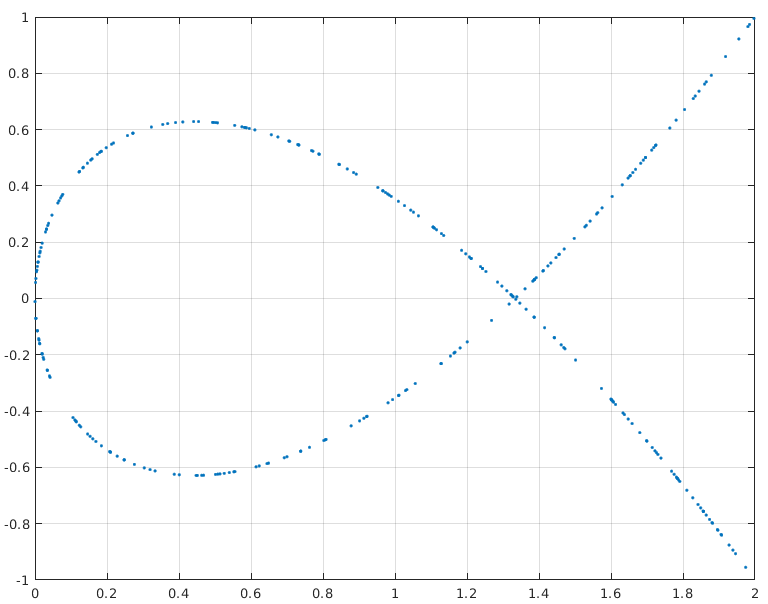}
	\caption{\rm A sample of 300 random points from the curve $\gamma: [-1,1] \rightarrow \bbR^2$ give by $\gamma(t)=(3t^3-2t, 2t^2)$ generated by our algorithm \nameref{alg:newsampler}. For more details on this example see~\ref{subsec:fig1example}. %From blue to yellow, from smaller to higher values on the z-axis. With red the $10$ points we sampled with our implementation of Alg.~\ref{alg:newsampler}.
	\label{fig:curve_sampled}}
\end{figure}

\subsection*{Notation}
For a real function $h : I\subseteq \mathbb{R} \rightarrow \mathbb{R}$, its $L^1$ norm is $\|h\|_{1}:=\int_\bbR\,|h(x)|\,ds$,
%\begin{equation}
%    \|h\|_{1}:=\int_\bbR\,|h(x)|\,ds,
%\end{equation}
where, by convention, we take $h(x)=0$ for $x\notin I$.
Its $L^\infty$-norm is defined as $\|h\|_\infty:=\sup_{x\in I}|h(x)|$.

Given a random variable $\fkx\in I$, we will write $\fkx\sim \varphi$ to indicate that $\fkx$ is distributed according to $\varphi$, i.e., $\bbP(\fkx\in J)=\int_{J}\varphi$. We will also denote by $\mcB$ the set of Borel subsets of $\bbR$. 

\subsection*{Organization of the paper}
In \S\ref{sec:efficientsampler} we define the total variation distance of two random variables and introduce the notion of an efficient sampler. In \S\ref{sec:tvdistance} we give bounds on the total variation distance. Inverse transform sampling is illustrated in \S\ref{sec:inv}; we also describe how the univariate solving (that inverse transform sampling requires) is done using the bisection method. %A discussion on how the univariate solving can be done using other methods completes this section.
In \S\ref{sec:sampling} we present our curve sampler, study its efficiency and analyze its complexity. We also include a review of Chebyshev approximations. In the end, we present our experimental results~in~\S\ref{sec:implementation} and our conclusions in \S\ref{sec:conclusions}.

\section{How good is a sampling method?\label{sec:efficientsampler}}

Given a continuous random variable $\fkx\in \bbR$ (\emph{target random variable}), we want to construct a sampler---an algorithm---whose output is identically distributed to $\fkx$. If the distribution of $\fkx$ is simple enough, the latter can be easily done. However, in general, we cannot sample $\fkx$ exactly and we can only obtain a random variable $\tilde{\fkx}\in \bbR$ (\emph{sampled random variable}) which behaves approximately like $\fkx$. Hence the following question arises: how well does the sampled random variable $\tilde{\fkx}$ approximate the target random variable $\fkx$? 

\subsection{Total Variation Distance}

The total variation distance measures how much the probabilities of two arbitrary events differ; the smaller the total variation distance is, the harder it is to distinguish the sampled random variable from the target random variable.

\begin{defi}~\cite{vempala2005geometric}
Let $\fkx,\tilde{\fkx}\in \bbR$ be random variables. The \emph{total variation distance} (\emph{TV distance}) of $\fkx$ and $\tilde{\fkx}$ is defined as
\[
\distTV(\fkx,\tilde{\fkx}):=\sup_{B\in\mcB}\left|\bbP(\fkx\in B)-\bbP(\tilde{\fkx}\in B)\right|,\]
where $\mcB$ is the set of Borel subsets of $\bbR$.
\end{defi}

\subsection{What is an efficient sampler? \label{sec:BSS}}

For a sampler to be efficient, we want it to run in time that is polylogarithmic in the error. We use the Blum-Shub-
Smale (BSS) model of computation \cite{BCSSbook}  to avoid problems arising from approximating continuous random variables with discrete ones. In the BSS model, real numbers can be stored exactly as a single unit during computations, and operations with real numbers are done at unit cost. We call \textit{BSS program}, a program, i.e., a finite list of commands, that can be implemented in a BSS machine. %Recall that a BSS machine is like a Turing machine but which can operate with real numbers and perform arithmetic operations and comparisons at cost one.

We introduce the notion of an \textit{efficient sampler}, which will be useful in measuring the performance of our sampling method.

\begin{defi}
\label{def:sampler}
Given a random variable $\fkx\in \bbR$, an \emph{efficient sampler for $\fkx$} is a pair of BSS programs $\textsc{S}:\bbN\times \bbR^k\times[0,1]^l\rightarrow \bbR$ and $\textsc{P}:\bbN\rightarrow \bbR^k$ such that: S1) on input $(\ell,x,u)$, the run-time of $\textsc{S}$ is at most $\mathrm{poly}(\ell)$, S2) on input $\ell$, the run-time of $\textsc{P}$ is at most $\mathrm{exp}(\ell)$, and S3) if $\fku\in[0,1]^l$ is uniformly distributed, then $\fkx_{\ell}:=\textsc{S}(\ell,\textsc{P}(\ell),\fku)$ is a random variable such that
\[
\distTV(\fkx_\ell,\fkx)\leq 2^{-\ell}.
\]
\end{defi}

%\begin{remark}[BSS model]
%By a BSS program, we mean that $\textsc{S}$ and $\textsc{P}$ can be implemented in a BSS machine, see~\cite{BCSSbook} for a definition. Recall that a BSS machine is like a Turing machine but which can operate with real numbers and perform arithmetic operations and comparisons at cost one.
%\end{remark}

The program $P$ in Def.~\ref{def:sampler} represents the preprocessing that the sampler requires in order to produce the correct result and $\fku$ the source of randomness. The output of $P$ is used together with the source of randomness $\fku$ as the inputs of the program $S$, which produces the sample. 

\begin{remark}[On-line and off-line computations]
In a more classical conception of an efficient sampler, the program $P$ would be missing. However, we have to take into account that samplers are not intended to be a run-once program, but they are intended to run many times. Because of this, it is reasonable to allow off-line computations---precomputations---, even if these are expensive, as long as we don't have to repeat them. In this way, our definition gives this possibility.
\end{remark}

\subsection{And if we have finite precision?}

If we have finite precision, then the sampled random variable $\tilde{\fkx}$ is discrete, and if this is the case, since the target random variable $\fkx$ is continuous, then
\[
\distTV(\fkx,\tilde{\fkx})=1.
\]
Hence, the TV distance does not allow us to evaluate how a discrete sampler approximates a continuous random variable directly.

A way around this problem is to turn the discrete approximation $\tilde{\fkx}$ into a continuous random variable by adding random noise to the values $\tilde{x_1},\ldots,\tilde{x_a}$ that $\tilde{\fkx}$. To do this, we sample 
\[
x_i+\fky_i,
\]
with probability $\bbP(\tilde{\fkx}=\tilde{x_i})$, where $\fky_1,\ldots,\fky_a$ some set of continuous random variables that we can sample. The resulting random variable will be continuous with density
\[
\sum_{i=1}^a \bbP(\tilde{\fkx}=\tilde{x_i})\delta_{\fky_i}(t-\tilde{x_i}).
\]
Hence, if the $\fky_i$ are simple enough, we have just specified a way of turning our discrete sample random variable into a continuous one for which we can evaluate the quality using the TV distance.

Unfortunately, analyzing in detail the precision needed goes beyond the scope of this paper. However, we note that all the methods produced in this paper---see the next two sections---are of this form and thus, in the future, we will perform a careful study of how the considered random samplers behave under finite precision.

\section{General bounds for the TV distance \label{sec:tvdistance}}
\
Let $I$ be a real interval. To bound the TV distance between two random variables $\fkx,\fky\in I$, we will use two methods: $L^1$-norms and interval partitions. The first one is used to approximate a random variable, by approximating its distribution. The second one approximates a random variable by approximating it in several intervals.

For the rest of this section, we use $\delta_{\fkx}$ and $\delta_{\fky}$ to denote the density functions of $\fkx$ and $\fky$ respectively. 

\subsection[Bounds using L1-norms]{Bounds using $L^1$-norms}

%Recall that the $L^1$ norm of a real function $\varphi:I\rightarrow \bbR$ is
%\begin{equation}
%    \|\varphi\|_{1}:=\int_\bbR\,|\varphi(x)|\,ds,
%\end{equation}
%where, by convention, we take $\varphi(x)=0$ for $x\notin I$. 
The following proposition shows the main technique that we apply for getting bounds on the TV distance.

\begin{prop}\label{prop:TVL1}
Let $\fkx,\fky\in\bbR$ be continuous random variables. Then
\begin{equation}
    \distTV(\fkx,\fky)\leq \left\|\delta_{\fkx}-\delta_{\fky}\right\|_{1}.
\end{equation}
\end{prop}
\begin{proof}
Let $B\in\mcB$. By definition, $\bbP(\fkx\in B)=\int_B\delta_{\fkx}$ and $\bbP(\fky\in B)=\int_B\delta_{\fky}$. Therefore, $|\bbP(\fkx\in B)-\bbP(\fky\in B)|=\left|\int_B(\delta_{\fkx}-\delta_{\fky})\right|\leq \left\|\delta_{\fkx}-\delta_{\fky}\right\|_{1}$.
\end{proof}

We recall that a usual way to bound the $L^1$-norm is to use the $L^\infty$-norm. For a function $h:I \subset \mathbb{R} \rightarrow \bbR$,  we have that
\begin{equation}\label{eq:L1Linfbound}
    \|h\|_1\leq\lambda(I)\|h\|_\infty
\end{equation}
where $\lambda$ is Lebesgue's measure---the length.

\subsection{Bounds using partitions}

The following proposition allows us to bound the total variation distance between $\fkx\in I$ and $\fky\in I$ using the information on how these random variables behave on a certain partition $\{J_i\}_{i=1}^k$ of $I$. 

\begin{defi}
Let $\fkx\in I$ be a continuous random variable and $J\subseteq I$. The \emph{restriction of $\fkx$ to $J$}, $\fkx_{|J}$, is the random variable whose density is given by $(\delta_{\fkx})_{|J}/\bbP(\fkx\in J)$, if $\bbP(\fkx\in J)\neq 0$, and by $1/\lambda(J)$, where $\lambda$ is the Lebesgue's measure, otherwise.
\end{defi}

One can see that to sample $\fkx_{|J}$ using $\fkx$ we only have to sample $\fkx$ until it lies on $J$ and, when this happens, output the sampled value. Note that this method requires on average $\bbP(\fkx\in J)^{-1}$ attempts, so the larger $\bbP(\fkx\in J)$ is, the more efficient this method becomes. 

The following proposition deals with the inverse problem: how do the errors in the partition of an interval add up?

\begin{prop}\label{prop:partitionsample}
Let $\fkx,\fky\in I$ be a random variables and $\{J_i\}_{i=1}^k$ be a partition of $I$. Then
\begin{multline}
        \distTV(\fkx,\fky)\leq \sum_{i=1}^k\bbP(\fkx\in J_i)\distTV(\fkx_{|J_i},\fky_{|J_i})\\+\sum_{i=1}^k|\bbP(\fkx\in J_i)-\bbP(\fky\in J_i)|
\end{multline}
\end{prop}
\begin{proof}
Assume, without loss of generality that for all $i$, $\bbP(\fkx\in J_i)$ and $\bbP(\fky\in J_i)$ are positive. Otherwise, the statement still holds, but the proof is slightly more convoluted. Fix $B\in\mcB$. By the definition of the $\fkx_{|J_i}$, $\bbP(\fkx\in B\cap J_i\mid \fkx\in J_i)=\bbP(\fkx_{|J_i}\in B)$. Therefore
\begin{equation}
    \bbP(\fkx\in B)=\sum_{i=1}^k\bbP(\fkx_{|J_i}\in B)\bbP(\fkx\in J_i),
\end{equation}
and so, after some elementary operations, we bound $|\bbP(\fkx\in B)-\bbP(\fky\in B)|$ by
\[
\sum_{i=1}^k|\bbP(\fkx_{|J_i}\in B)-\bbP(\fky_{|J_i}\in B)|\bbP(\fkx\in J_i)+\sum_{i=1}^k|\bbP(\fkx\in J_i)-\bbP(\fky\in J_i)|.
\]
Now, by maximizing over $B\in\mcB$, we conclude.
\end{proof}

The above proposition suggests the strategy of partition sampling (Algorithm~\nameref{alg:partitionsampler}). In other words, to sample $\fkx$, we only need to sample the $\fkx_{|J_i}$ and to compute the probabilities $\bbP(\fkx\in J_i)$ with enough precision for some partition $\{J_i\}_{i=1}^k$ of $I$.

\begin{algorithm}
\DontPrintSemicolon
\SetKwInOut{input}{Input}
\SetKwInOut{output}{Output}
\caption{\textsc{PartitionSampler}}\label{alg:partitionsampler}

\input{Partition $J_1,\ldots,J_k$ of $I$\\
$\tilde{p}\in\Delta^{k-1}:=\{p\in\bbR^k_\geq 0\mid \|p\|_1=1\}$\\
Approximate samplers $\textsc{S}_i$ for $\fkx_{|J_i}$}
\output{Approximate sample of $\fkx\in I$}
Sample $i\in\{1,\ldots,k\}$ with probability $\tilde{p_i}$\;
$\fkx\leftarrow \textsc{S}_i$\Comment*[r]{We use sampler $\textsc{S}_i$ to get random $\fkx\in J_i$}
Output $\fkx$\;
\end{algorithm}

We omit the formal proof that this procedure gives an efficient sampler if the $\textsc{S}_i$ are efficient samplers. To see this, we only have to note that this method's run-time will be at most the run-time of the $\textsc{S}_i$. As we can precompute the probabilities $\bbP(\fkx\in J_i)$, we can choose sufficiently good values for $\tilde{p}$ within the required restrictions for almost all cases we consider here.

\section{Inverse Transform Sampling by Bisection\label{sec:inv}}

Let $I=[a,b]\subset \mathbb{R}$ and $\varphi:I\rightarrow [0,\infty)$ be a density function. Inverse transform sampling is based on the fact that the solution $\tilde{\fkx}\in I$ of $\int_{a}^{\tilde{\fkx}}\varphi(s)\,\mathrm{d}s=\fku$, for $\fku\in[0,1]$ uniformly distributed, has density $\varphi$ (see Alg.\ref{alg:inversetransformsampler} for its pseudocode).

\begin{algorithm}
\DontPrintSemicolon
\SetKwInOut{input}{Input}
\SetKwInOut{output}{Output}
\caption{\textsc{InverseTransformSampler}}\label{alg:inversetransformsampler}

\input{$I=[a,b]$, $\varphi:I\rightarrow [0,\infty)$ such that $\int_I\varphi(t)\,\mathrm{d}t=1$}
\output{$\fkx\sim \varphi$}
Sample $\fku\in [0,1]$ uniformly\;
Find the solution $\fkx$ of $\int_{a}^\fkx\varphi(s)\,\mathrm{d}s=\fku$\;
Output $\fkx$\;
\end{algorithm}

However, any reader of this pseudocode will be suspicious about all the details swept under the rug in line 2. How do we solve $\int_{a}^\fkx\varphi(s)\,\mathrm{d}s=\fku$? And how fast can we do it? Even though this is an important question regarding the complexity of sampling, we feel that it is unaddressed by the literature, so we discuss it.

Let the cumulative distribution function corresponding to the density $\varphi$ be
\begin{equation}
    \Phi(x):=\int_a^x\,\varphi(s)\,\mathrm{d}s.
\end{equation}
We want to solve the equation
\begin{equation}\label{eq:invtransequation}
    \Phi(x)=u\in[0,1].
\end{equation}
When solving this equation, the bisection method outputs an interval containing the root. This interval is found by repeatedly subdividing the initial interval, and selecting the one such that $u - \Phi $ has different signs at the two endpoints. Let $\ell \in \mathbb{N}$; we stop subdividing after $\ell$ iterations. In the end, we choose a point at random in the final interval.
We integrate the bisection method in inverse transform sampling in the algorithm \nameref{alg:bisectionsampler}.

\begin{algorithm}
\DontPrintSemicolon
\SetKwInOut{input}{Input}
\SetKwInOut{output}{Output}
\caption{\textsc{BisectionSampler}}\label{alg:bisectionsampler}
\input{$\varphi:[a,b]\rightarrow [0,\infty)$ such that $\int_0^1\varphi(t)\,\mathrm{d}t=1$\\
$\ell\in\bbN$}
\output{Approximate sample $\fkx$ for $\varphi$}
Sample $\fku\in [0,1]$ uniformly, $i\gets 0$\;
$x_l\gets a$, $v_l\gets \mathrm{sign}\left(\fku -\Phi(x_l)\right)$, 
$x_r\gets b$, $v_r\gets \mathrm{sign}\left(\fku -\Phi(x_r)\right)$\;
\While{$i<\ell$}{
$x_m\gets (x_l+x_r)/2$, $v_m\gets \mathrm{sign}\left(\fku -\Phi(x_m)\right)$\;
\eIf{$v_l=v_m$}{
$x_r\gets x_m$, 
$v_r\gets v_m$, 
$i\gets i+1$\;
}{
$x_l\gets x_m$, 
$v_l\gets v_m$, 
$i\gets i+1$\;
}
}
Sample $\tilde{\fkx}\in [x_l,x_r]$ uniformly\;
Output $\tilde{\fkx}$\;
\end{algorithm}

The following theorem shows that \nameref{alg:bisectionsampler} produces a nice sampler.

\begin{theo}\label{theo:bisection}
Let $\varphi:[a,b]\rightarrow [0,\infty)$, $\fkx\sim\varphi$ and $\tilde{\fkx}_\ell$ the output of \nameref{alg:bisectionsampler} for $\ell\in\bbN$. Then
\begin{equation}
    \distTV(\fkx,\tilde{\fkx}_\ell)\leq 2^{-\ell}|b-a|\max_{x\in[a,b]}|\varphi'(x)|.
\end{equation}
\end{theo}
\begin{proof}
Without loss of generality, assume that $a=0$. For $k\in\{0,\ldots,2^{\ell}-1\}$, let $J_k:=[bk/2^\ell,b(k+1)/2^\ell]$. The $\tilde{\fkx}_\ell$ produced by \nameref{alg:bisectionsampler} can also be produced as follows:
\begin{enumerate}
    \item Choose the interval $J_k$ at random with probability $\bbP(\fkx\in J_k)$. Note that the solution of~\eqref{eq:invtransequation} lies in $J_k$ with probability $\bbP(\fkx\in J_k)$.
    \item Sample $\fku_k\in J_k$ uniformly.
\end{enumerate}
Therefore, by Proposition~\ref{prop:partitionsample}, to compute $\distTV(\fkx,\tilde{\fkx}_\ell)$, we only have compute the TV distance between $\fkx_{|J_k}$ and the uniformly distributed $\fku_k\in J_k$. Now, after an elementary computation and Proposition~\ref{prop:TVL1},
\begin{equation}
    \distTV(\fkx_{|J_k},\fku_k) \le \bbP(\fkx\in J_k)^{-1}\|\varphi_{|J_k}-2^\ell\bbP(\fkx\in J_k)\|_1. 
\end{equation}
Now, by the mean value theorem, $\bbP(\fkx\in J_k)=b\varphi(\xi_k)/2^\ell$ for some $\xi_k\in J_k$, so we obtain
\begin{equation}
    \distTV(\fkx_{|J_k},\fku_k) \le \bbP(\fkx\in J_k)^{-1}\|\varphi_{|J_k}-\varphi(\xi_k)\|_1. 
\end{equation}
Now, $|\varphi(s)-\varphi(\xi_k)|\leq b2^{-\ell}\max_{x\in J_k}|\varphi'(x)|$, by the mean value theorem, and so $\distTV(\fkx_{|J_k},\fku_k)\leq \bbP(\fkx\in J_k)^{-1}b2^{-2\ell}\max_{x\in J_k}|\varphi'(x)|$, concluding the proof.
\end{proof}

Note that \nameref{alg:bisectionsampler} needs to perform a minimum of
\[
\max\left\{0,\sup\{\log |\varphi'(x)|\mid x\in [a,b]\}\right\}
\]
iterations. Moreover, note that in the bisection method, we can interchange precomputation and computation with no effect to our notion of the efficient sampler.

\section[Sampling points on a curve]{Sampling points from a curve \label{sec:sampling}}

Let $\gamma:I:=[-1,1] \rightarrow \mathbb{R}^n$ be the parametrization of a real polynomial curve of degree $d$---after a linear change of coordinates we can always assume that $I=[-1,1]$. Since we want to generate random points $\fkx\in\gamma(I)$ uniformly with respect to the arc-length, we only need to sample a random parameter $\fkt\in I$ distributed according to the normalized speed
\begin{equation}\label{eq:normspeed}
    \varphi(t):=\left(\int_{-1}^1\|\gamma'(s)\|_2\mathrm{d}s\right)^{-1}\|\gamma'(t)\|_2
\end{equation}
and then take the random variable $\gamma(\fkt)\in\gamma(I)$ which will have the desired distribution.

When $n\geq 2$, we have that $\varphi$ is not a polynomial. Because of this, to perform the inverse transform sampling, even by bisection, we will approximate $\varphi$ by a Chebyshev approximation $\tilde\varphi$ for which computing
\[
\tilde\Phi(t):=\int_{-1}^t\tilde\varphi(s)\,\mathrm{d}s
\]
is a lot easier than computing
\[
\Phi(t):=\int_{-1}^t\varphi(s)\,\mathrm{d}s.
\]
Now, to make the Chebyshev approximation faster, we will split the interval into subintervals.

First, we review the Chebyshev approximation; then, we apply it to the case of interest; finally, we show how splitting accelerates the Chebyshev approximation. The algorithm appears in Algorithm~\ref{alg:newsampler}. Again, we observe that this algorithm is very similar to the one proposed in~\cite{Olver13}, but our main contribution is not the sampler itself but the error analysis using the total variation distance.

\subsection{Chebyshev approximations\label{sec:chebyshev}}

We follow mainly~\cite{mason_chebyshev_2003} and~\cite{Trefethen2013}. Recall that the \emph{$k$th Chebyshev polynomial} is the polynomial given by
\begin{equation}
      \chev_k(X) =\sum_{i=0}^{k}\binom{k}{2i}(1-X^2)^iX^{k-2i}, 
\end{equation}
where $\chev$ is the first initial of Chebyshev in the Cyrillic script. Alternatively, note that $\chev_k$ satisfies
\begin{equation}
      \chev_k(x) =\cos(k\arccos(x)). 
\end{equation}
for $x\in I$. As a consequence, the $k$ zeros of $\chev_k$ are given by
\begin{equation}
    \zeta_{a,k}:=\cos\left(\frac{(1+2a)\pi}{2k}\right)
\end{equation}
with $a\in \{0,\ldots,k-1\}$.

\subsubsection{Chebyshev interpolation}

The \emph{$k$th Chebyshev interpolant} of $f:I\rightarrow\bbR$ is the unique degree $k$ polynomial $\chevI_k(f)$ satisfying for $a\in\{0,\ldots,k\}$,
\begin{equation}
    \chevI_k(f)(\zeta_{a,k+1})=f(\zeta_{a,k+1}).
\end{equation}
To compute $\chevI_k(f)$, there is no need to solve the system above thanks to the following proposition.

\begin{prop}\label{prop:interpolationcomputation}\cite[Thm.~6.7]{mason_chebyshev_2003}
Let 
\[
\chevI_k(f)=\frac{c_0}{2} + \sum_{a=1}^k c_a \chev_a
\]
be the $k$th Chebyshev interpolant of $f:I\rightarrow\bbR$. Then
\begin{equation*}\tag*{\qed}
    c_a = \frac{2}{k+1} \sum_{i=0}^{k}f(\zeta_{i,k}) \chev_a(\zeta_{i,k}).
\end{equation*}
\end{prop}

Let us remind that $\chevI_k(f)$ is not equivalent to truncating the Chebyshev series up to degree $k$~\cite[Ch.~4]{Trefethen2013}.

\subsubsection{Evaluation of Chebyshev interpolants}

Given a Chebyshev interpolant $\chevI(f)$, we can expand it in the monomial basis and then evaluate it using, for example, Ruffini-Horner's method. However, we have a version of Ruffini-Horner's method that works directly for Chebyshev expansions of a polynomial.

\begin{prop}\cite[pp.~55-56]{FoxParker}\label{prop:evalchebyshev}
Let $p=\sum_{a=0}^kc_a\chev_k$, then for every $x\in\bbR$,
\begin{equation}
p(x) = \frac{1}{2} (\bcomp_0(x)-\bcomp_2(x))  
\end{equation}
where $\bcomp_0(x)$ and $\bcomp_2(x)$ are computed through the following back\-wards-recursive relation 
\begin{equation*}\tag*{\qed}
\begin{cases}
        \bcomp_{k+1}(x) &= \bcomp_{k+2}(x)=0\\
        \bcomp_a (x) &= 2x\bcomp_{a+1}(x)-\bcomp_{a+2}(x)+c_a.

\end{cases}   
\end{equation*}
\end{prop}

\subsubsection{Speed of convergence}

To estimate the error of the Chebyshev approximation, we consider the so-called \emph{Bernstein ellipse} given by
\[
E_\rho := \{z\in \mathbb{C} \mid |z+\sqrt{z^2-1}| = \rho \},
\]
where $\rho> 1$. Note that $E_\rho$ is the ellipse with foci $-1$ and $1$ and focal distance $\rho+\rho^{-1}$. We have the following theorem.

\begin{theo}\label{theo:approx_error}\cite[Theorem~8.2]{Trefethen2013}.
If $f$ is analytic on the elliptic disc given by $E_{\rho}$, then
\begin{equation}
    \norm{f-\chevI_k(f)}_\infty \le \frac{4\|f\|_{E_\rho}{\rho}^{-k}}{\rho-1},
\end{equation}
where $\|f\|_{E_\rho}: = \max_{z\in E_{\rho}} |f(z)|$. \eproof
\end{theo}

To obtain theoretical bounds of $\|f\|_{E_\rho}$, the following inequality by Bernstein will make our job easier.

\begin{theo}\label{theo:bernsteinineq}\cite[Exercise~8.6]{Trefethen2013}
Let $f$ be a polynomial of degree $d$ and $\rho>1$. Then
\begin{equation*}\tag*{\qed}
    \|f\|_{E_\rho}\leq \rho^d\|f\|_\infty.
\end{equation*}
\end{theo}

Note that this inequality does not serve to approximate polynomials by a Chebyshev interpolant of a lower degree. However, we will use it to control the quantity of interest for $\varphi$, i.e., the normalized speed of $\gamma$~\eqref{eq:normspeed}. The hard part will be estimating a sufficiently small $\rho$ so that $\varphi$ admits an analytic extension to the interior of $E_\rho$.

\subsubsection{Integration formulas}

Imagine we want to compute the integral (definite or indefinite) of $\chevI_k(f)$. To do this, we use the following proposition.

\begin{prop}\label{prop:interpolationintegrals}\cite[pp.~54-55]{FoxParker}\cite[pp.~45-46, 59]{mason_chebyshev_2003}
Let 
\[
p=\frac{c_0}{2} + \sum_{a=1}^k c_a \chev_a.
\]
Then
\[
\sum_{a=1}^{k+1}pc_a\chev_a
\]
with
\[
 pc_a = 
 \begin{cases}
     \frac{c_{a-1}-c_{a+1}}{2a}, &\text{if }a=1,\dots, k-1\\
   \frac{c_{k-1}}{2k}, &\text{if }a=k\\
    \frac{c_{k}}{2(k+1)}, &\text{if }a=k+1
 \end{cases}
\]
is a primitive function of $p$.

Moreover,
\begin{equation*}\tag*{\qed}
    \int_{-1}^1 p(x)\,\mathrm{d}x = c_0 - \sum_{a=2}^n\left(\frac{1+(-1)^a}{a^2-1}\right) c_a.
\end{equation*}
\end{prop}

The following fact will be useful later. Let $F(t):=\int_{-1}^tf(s)\mathrm{d}s$ for $t\in[-1,1]$, then
\[
\left\|F\right\|_\infty\leq \|f\|_1\leq 2\|f\|_\infty.
\]
Hence, Theorem~\ref{theo:approx_error} allows us to also control the error of the integral approximation.

\subsection{Bernstein ellipse for the speed}

The following theorem shows that the conditions of Theorem~\ref{theo:approx_error} are satisfied for the speed $\varphi$ and gives possible $\rho$s we can take.

\begin{theo}\label{theo:approx_error_speed}
Let $\varphi$ be given as in \eqref{eq:normspeed}, non-vanishing in $I$. Let $\rho>1$ be such that $\rho$ does not exceed
\begin{equation*}
    \rho^*(\gamma):=\min_{\substack{z\in\bbC\\ \|\gamma'(z)\|_2=0}}\left\{\frac{|z+1|+|z-1|+ \sqrt{(|z+1|+|z-1|)^2 -4}}{2}\right\},
\end{equation*}
then $\varphi$ admits an analytic extension $\varphi_{\mathrm{an}}$ to the interior of $E_\rho$ and
\[
\|\varphi_{\mathrm{an}}\|_{E_\rho}\leq \rho^d\|\varphi\|_\infty.
\]
\end{theo}
\begin{proof}
Note that $\varphi$ is the square root of the polynomial $\|\gamma'\|^2$. To analytically extend such a function to the interior of the ellipse $E_\rho$, we need that no complex root $z$ of $\|\gamma'\|^2$ lies inside $E_\rho$. Now, $z$ lies inside the interior of $E_\rho$ if and only if
\[
\rho \geq \frac{|z+1|+|z-1|+ \sqrt{(|z+1|+|z-1|)^2 -4}}{2}.
\]
Recall that $E_\rho$ is the ellipse with foci $-1$ and $1$ and focal distance $\rho+\rho^{-1}$, so $z$ lies in its interior if and only if $|z-1|+|z+1|\leq \rho+\rho^{-1}$. The latter is equivalent to the above inequality for $\rho$.

Since $\varphi_{\mathrm{an}}^2$ is a polynomial of degree $2d$ on $I$, it is a polynomial of degree $2d$ on the interior of $E_\rho$. Hence, by Bernstein's inequality (Theorem~\ref{theo:bernsteinineq}),
\[
\|\varphi_{\mathrm{an}}^2\|_{E_\rho}\leq \rho^{2d}\|\varphi_{\mathrm{an}}^2\|_\infty=\rho^{2d}\|\varphi^2\|_\infty.
\]
Since the square root of the maximum is the maximum of the square root, the desired bound follows. 
\end{proof}

\begin{remark}
Note that Theorem~\ref{theo:approx_error_speed} gives a conservative bound for $\|\varphi\|_{E_\rho}$, which we use for giving an upper bound for the complexity. Since we can precompute this quantity, we do this in the off-line part of \nameref{alg:newsampler}, so that we get better run times.
\end{remark}

Observe that $\rho^*(\gamma)$ is optimal. In this way, if we want a Chebyshev interpolant $\chevI_k(\varphi)$ such that $\|\chevI_k(\varphi)-\varphi\|_\infty\leq\varepsilon$, then the degree of this interpolant has to satisfy
\begin{equation}\label{eq:degbound}
    k \geq \frac{1}{\log\rho^*(\gamma)}\left(\ln\frac{1}{\varepsilon}+\log\|\varphi\|_{E_{\rho^*(\gamma)}}+2-\log(\rho^*(\gamma)-1) \right).
\end{equation}
by Theorem~\ref{theo:approx_error}.

We now give theoretical bounds for $\rho^*$. But before, let us define the geometric parameter that will appear in these bounds.

\begin{defi}\label{defi:conditiongammasampling}
Let $\gamma:I\rightarrow \bbR^n$ be the parametrization of a real polynomial curve of degree $d$ such that the polynomial in the $i$th component, $\gamma_i$, is given by
\[\gamma_i=\sum_j\gamma_{i,j}T^j.\]
Then the \emph{condition number for sampling $\gamma$} is quantity 
\begin{equation}\label{eq:conditiongammasampling}
   \cond(\gamma):=\frac{\|\gamma'\|_{o}}{\inf_{t\in I}\|\gamma'(t)\|_2}\in [1,\infty]
\end{equation}
where $\|\gamma'\|_{o}=\sum_{i=1}^n\sum_{j}j|\gamma_{i,j}|$ is the sum of the absolute value of all coefficients of $\gamma$ multiplied each by the degree of their term.
\end{defi}

\begin{remark}
The condition number for sampling $\gamma$, $\cond(\gamma)$, is finite as long as $\varphi$ is non-vanishing in $I$. Note that the idea is that the nearer is $\gamma$ to have zero speed at a point, the harder it is to sample a random point in it.
\end{remark}

\begin{theo}\label{theo:boundsrho}
Let $\gamma:I\rightarrow \bbR^n$ be the parametrization of a real polynomial curve of degree $d$ with non-vanishing speed. Then
\[
\rho^*(\gamma)\geq 1+\frac{1}{\enumber \cdot d\cdot \cond(\gamma)}.
\]
\end{theo}
\begin{proof}
Since $\varphi$ is non-vanishing, $\cond(\gamma)< \infty$. Let $\varepsilon\in(0,1/d)$ and consider $I_\varepsilon:=\{z\in\bbC\mid \dist(z,I)\leq\varepsilon\}$. Now, by~\cite[Proposition 3.6]{TCTcubeI-journal}, for each $i$, the map
\[
I_\varepsilon\ni z \mapsto |\gamma_i'(z)|/\|\gamma_i'\|_o
\]
is $(\enumber \cdot d)$-Lipschitz. Hence, the map
\[
I_\varepsilon\ni z \mapsto \|\gamma'(z)\|_2/\|\gamma'\|_o.
\]
is $(\enumber \cdot d)$-Lipschitz. In this way, if $\varepsilon=1/(\enumber \cdot d\cdot \cond(\gamma))$, we have that $\|\gamma'\|^2$ does not have zeros inside $I_\varepsilon$. Now, if for some $\rho\geq 1$, $E_\rho\subseteq I_\varepsilon$, then $\rho\leq\rho^*(\gamma)$.

By the definition of $E_\rho$, $E_\rho\subseteq I_\varepsilon$ if and only if a) $(\rho+\rho^{-1})/2-1\leq \varepsilon$ (major semiaxis bound) and b) $\sqrt{\rho^2+\rho^{-2}-2}/2\leq \varepsilon$ (minor semiaxis bound). 

Now, $\rho=1+\varepsilon$, with $\varepsilon=1/(\enumber \cdot d\cdot \cond(\gamma))\leq 1$, satisfies these inequalities. Hence the claim follows.
\end{proof}

\subsection{Acceleration through splitting the curve}\label{subsec:splittinginterval}

Whenever we split the curve, we should expect the value of $\rho^*(\gamma)$ to increase. The reason for this is that, after renormalization of a smaller interval, the zeros of $\gamma'$ are further away and so the value of $\rho$ should increase. However, we observe that this increase will depend on the value of $\rho^*(\gamma)$, so there is not a uniform constant factor improvement independent of $\rho^*(\gamma)$.

The above paragraph suggests that we can just perform a fixed number of binary subdivision steps to accelerate the algorithm. This strategy does indeed accelerate the sampler, as shown by experiments (see subsection~\ref{subsec:splittingintervalimp} and Figure~\ref{fig:splitting}).

To conclude, note that as we compute $\rho^*(\gamma)$, we have to compute also the complex roots of $\gamma'$. Therefore we can split the interval $I$ along the points
\[
\Re z_1,\ldots,\Re z_d
\]
where $z_1,\overline{z_1},\ldots,z_d,\overline{z_d}$ are the complex roots of $\gamma'$. This subdivision accelerates the algorithm significantly as it forces the roots of $\gamma'$ to lie on the endpoints of each interval or outside.

\subsection{The sampler}

We give the sampler for the curve $\gamma:I\rightarrow \bbR^n$, \nameref{alg:newsampler}, without indicating the subdivision procedure. We observe that excepting the last call, the sampler performs the majority of its operations off-line, so they don't have to be repeated in each call.

\begin{algorithm}
\DontPrintSemicolon
\SetKwInput{input}{Input}
\SetKwInput{output}{Output}
\caption{\textsc{CurveSampler}}\label{alg:newsampler}
\input{$\gamma:I\rightarrow \mathbb{R}^n$ of degree $d$\\
$\ell\in\bbN$}
\output{Approximate sample $\fkt$ of $\varphi:=\norm{\gamma'}_2/\int_{-1}^1\norm{\gamma'}_2(s)\mathrm{d}s$}
%\postcondition{Some condition of nearby}
$\varphi\gets\norm{\gamma'}_2/\int_{-1}^1\norm{\gamma'}_2(s)\mathrm{d}s$\Comment*[r]{off-line}
$Z\gets \{z\in \bbC\mid \gamma'(z)=0\}$\Comment*[r]{off-line}
$\rho^*\gets\min_{z\in Z}\frac{|z+1|+|z-1|+ \sqrt{(|z+1|+|z-1|)^2 -4}}{2}$\Comment*[r]{off-line}
$M\gets \|\varphi\|_{E_{\rho^*}}$\Comment*[r]{off-line}
$k\gets 5+\ell+\lceil (\log M-\log(\rho^*-1))/\log\rho^*\rceil$ \Comment*[r]{off-line}
$\tilde{\varphi}\gets\chevI_k(\varphi)/\int_{-1}^1\chevI_k(\varphi)(s)\mathrm{d}s$\Comment*[r]{off-line}
$\tilde{\Phi}\gets$Integral of $\chevI_k(\varphi)$\Comment*[r]{off-line}
$\ell_B\gets 1+\ell+\max\{0,\log\|\tilde{\varphi}'\|_\infty\}$\Comment*[r]{off-line}
$\fkt\gets$\nameref{alg:bisectionsampler}$(\tilde{\varphi},\ell_B)$\;
Output $\fkt$\;
\end{algorithm}

\subsection[Complexity of CurveSampler]{Complexity of \nameref{alg:newsampler}}

Recall that we are working in the BSS model with square roots, so we assume that we can evaluate $\|\gamma'(t)\|_2$ exactly. Our main theorem is the following one.

\begin{theo}\label{thm:main}
Let $\gamma :I \rightarrow \mathbb{R}^n$ be a polynomial parameterized curve. \nameref{alg:newsampler} is an efficient sampler for $\fkt\in I$ uniformly distributed with respect to the normalized speed of $\gamma$. Moreover, it performs
\[
\Oh(\ell^2(1+\log\cond(\gamma))^2\cond(\gamma)^2)
\]
off-line arithmetic operations, where $\Oh$ has constants depending on the degree of $\gamma$; and
\[
\Oh(\ell^3(1+\log d\cond(\gamma))^3d^3\cond(\gamma)^3)
\]
on-line arithmetic operations to achieve an error of $2^{-\ell}$ in the TV distance. 
\end{theo}

\begin{remark}
Even though we are ignoring the complexity of the offline part---many of those parts can be done in $\text{poly}(d)$ arithmetic operations up to the desired degree of precision. This is why we focus on the dependence on the error.
\end{remark}

\begin{proof}[Proof of Theorem~\ref{thm:main}]
 On the one hand, by Theorem~\ref{theo:approx_error},
\begin{equation}
    \|\tilde\varphi-\varphi\|_1\leq 2\|\chevI_k(\varphi)-\varphi\|_1\leq \frac{16M}{\rho^*-1}(\rho^*)^{-k}\leq 2^{-(1+\ell)}
\end{equation}
where the first inequality follows from $\|\varphi\|_1=1$, and so if $\fks\sim\varphi$, then, by Proposition~\ref{prop:TVL1},
\begin{equation}
   \distTV(\fks,\tilde\fkt)\leq 2^{-(1+\ell)}
\end{equation}
where $\fkt\sim\tilde\varphi$. On the other hand, by Theorem~\ref{theo:bisection}, 
\[
\distTV(\tilde\fkt,\fkt)\leq 2^{1-\ell_B}\|\tilde\varphi'\|_\infty\leq 2^{-(1+\ell)}.
\]
Hence,
\[
\distTV(\fks,\fkt)\leq 2^{-\ell}
\]
and so to show that \nameref{alg:newsampler} is an efficient sampler for $\varphi$, we only need to bound the complexity as desired.

For the off-line part, we need to bound the number of arithmetic operations in terms of $\ell$. In line 2, we use some solver; for line 3, we just minimize over the roots found in line 2; for line 4, we only have to parameterize the boundary of $E_{\rho^*}$ and find the minimum, due to the maximum modulus principle; for line 5, we do the assignment; for line 6, we use Proposition~\ref{prop:interpolationcomputation} and the second part of Proposition~\ref{prop:interpolationintegrals}; for line 7, we use the first part of Proposition~\ref{prop:interpolationintegrals}. By observing these, we see that the number of arithmetic operations is at most $\Oh(k^2)$, which by the definition of $k$ in line 5 and Theorems~\ref{theo:approx_error_speed} and \ref{theo:boundsrho} transforms to
\[
\Oh(\ell^2(1+\log d\cond(\gamma))^2d^2\cond(\gamma)^2)
\]
where we use that $1/\log(1+x)\leq 2/x$.

For the on-line part, we perform $\ell_B$ evaluations of $\tilde{\Phi}$, taking each evaluation, by Proposition~\ref{prop:evalchebyshev}, $\Oh(k)$ operation. Thus we perform $\Oh(\ell_Bk)$ arithmetic operation. We have to bound $\ell_B$ now. However, this is equivalent to bounding
$
\|\tilde{\varphi}'\|_\infty
$.

Using the theory of Chebyshev polynomials~\cite{Trefethen2013}, we have that
\[
\|\chevI_k(\varphi)'-\varphi'\|_\infty\leq 16Mk^2(\rho^*)^{3-k}/((\rho^*)-1).
\]
Thus,
\[
\|\tilde{\varphi}'\|_\infty\leq \|\varphi'\|_\infty+\frac{16Mk^2\rho^{3-k}}{\rho-1}\leq \|\varphi'\|_\infty+4k^2(\rho^*)^{2-\ell}.
\]
Now, on the one hand,
\[
4k^2(\rho^*)^{2-\ell}\leq \Oh(\ell^2(1+\log d\cond(\gamma))^2d^2\cond(\gamma)^2);
\]
and on the other hand,
\[
|\varphi'|=\langle \gamma',\gamma''\rangle/\|\gamma'\|^2_2\leq\|\gamma''\|_2/\|\gamma'\|_2 
\]
where $\|\gamma''\|_2\leq d\|\gamma'\|_o$ by~\cite[Proposition 3.6]{TCTcubeI-journal}. Thus
$
    \|\varphi'\|_\infty\leq d\cond(\gamma)
$.
Putting this together, we obtain, the bound
$\ell_B\leq \Oh(\ell+\log d\cond(\gamma)+\ell^2(1+\log d\cond(\gamma))^2d^2\cond(\gamma)^2)$. Hence we are done.
\end{proof}

\section{Implementation and experiments \label{sec:implementation}}

We provide an open-source implementation in {\tt Matlab} of the studied method that can be accessed at
\begin{center}
    \href{https://github.com/TolisChal/sampling_curves}{\tt github.com/TolisChal/sampling\_curves}
\end{center}
It is an original implementation, up to technical modifications\footnote{For example, we allow $2^\ell$ to be any real number and not only a power of $2$.}, of \nameref{alg:newsampler} using \nameref{alg:bisectionsampler} which allows us to sample random points from a given parametric polynomial curve uniformly with respect to the arc-length.

Our implementation relies on a few standard routines from {\tt Matlab}'s toolbox. In particular, we use (i) {\tt chebyshevT()} to evaluate the $k$th degree Chebyshev polynomial $\chev_k(x)$, (ii) {\tt roots()} to compute the zero set of the speed $\norm{\gamma'(t)}^2_2$ for the computation of $\rho^*$ (line 3 in~\nameref{alg:newsampler}), and (iii) {\tt fmincon()} to solve the optimization problem required to compute $\|f\|_{E_\rho}$ (line 4). All computations were performed on a PC with {\tt Intel\textregistered\ Pentium(R) CPU G4400 @ 3.30GHz $\times$ 2 CPU} and {\tt 16GB RAM}.

\subsection{The example of Figure~\ref{fig:curve_sampled}}\label{subsec:fig1example}

In the example of Figure~\ref{fig:curve_sampled}, we execute \nameref{alg:newsampler} to produce a sample of 300 random points. During the execution, we can see that the speed of the curve $\gamma:[-1,1]\ni t\mapsto (3t^2-2t,2t^2)$,
\[
\|\gamma'(t)\|_2=\sqrt{42t^2-24t+4},
\]
is approximated by a Chebyshev approximation of degree $35$. In this example, we took $\ell=4$, so that $2^{-\ell}<0.1$.

\subsection{Random experiments: Table~\ref{tab:experiments}}

In Table~\ref{tab:experiments}, we show the results of performing several random experiments. For these experiments, we consider random parametric polynomial curves $\gamma:[-1,1]\rightarrow \bbR^n$ of degree $d\in\{5,10,15,20\}$ with $n\in \{20,40,60,80,100\}$ components and with error specifications $2^{\ell}\in\{10,100\}$. The coefficients of the polynomials in $\gamma$ are independent, identically distributed standard Gaussian random variables. In the table, we display the degree of the Chebyshev interpolant ($k$), the preprocess time ($Pr.\ T$) in  seconds and the time per generated sample after preprocessing ($T/s$) in seconds.

We can see that the degree $k$ of the Chebyshev interpolant increases with the degree of $d$ of the polynomial curve. The run-time of the preprocessing takes a few seconds, while the time per sample after preprocessing is smaller than $1$ second for every instance. 

\subsection{A curve of degree 10: Figures~\ref{fig:degree_error} and \ref{fig:splitting}}\label{subsec:splittingintervalimp}

In Figures~\ref{fig:degree_error} and \ref{fig:splitting}, we consider a parametric polynomial curve $\gamma:[-1,1]\rightarrow \bbR^{50}$ of degree $10$. In both cases, we plot the degree of the used Chebyshev interpolant (k) with respect to the inverse of the error $2^{\ell}$.

In Figure~\ref{fig:degree_error}, we simply show the evolution of the degree. In Figure~\ref{subsec:splittinginterval}, we illustrate the theoretical discussion of subsection~\ref{subsec:splittinginterval} by plotting the degree used when we don't split the interval (blue line) against the maximum degree used when we split the interval into four sub-intervals. This shows that a few subdivisions can significantly reduce the degree of the used Chebyshev interpolants.

\subsection{A polynomial density: Figures~\ref{fig:degree_time} and~\ref{fig:degree_degree}}

We consider the curve $\gamma:[-1,1]\ni t\mapsto (1+T+\cdots+T^d)(1,1,1)\in \bbR^3$ to have an explicit example where we can see the evolution of the algorithm with respect the degree $d$. In Figure~\ref{fig:degree_time}, we can see the the run-time per sample after preprocessing against $d$; and, in Figure~\ref{fig:degree_degree}, we can see the degree of the Chebyshev interpolant ($k$) against the degree. Interestingly, both increase sub-linearly with the degree $d$.

\section{Conclusions}\label{sec:conclusions}

In this paper, we initiated the study of the errors in random sampling in the context of algebraic geometry by studying a sampling method in the context of parametric polynomial curves. More precisely, we show that the method in~\cite{Olver13} is efficient for generating random point on a parametric polynomial curve both theoretically (see Theorem~\ref{thm:main}) and in practice (see \S\ref{sec:implementation}). However, this is just the first step towards obtaining error bounds for the methods generating random points in algebraic varieties---needed for theoretical guarantees for applications of TDA in algebraic geometry~\cite{breidingkalisniksturmfelsweinstein2018}.

Interestingly, the experiments suggest that the considered method might be faster than what our theoretical estimates suggest. This discrepancy might be because the bounds in Theorem~\ref{theo:approx_error_speed} or \ref{theo:boundsrho} are too pessimistic. To improve these in the future, we might need a sharper definition of the condition number $\cond(\gamma)$ or substituting some of the inequalities---especially those in which the degree appears---by inequalities that adapt better to the geometry of each curve. Thus, we feel that further theoretical work is needed to fully understand~\nameref{alg:newsampler}. We also note that alternative strategies to \nameref{alg:bisectionsampler}, for examples, those using Newton's method, need to be analyzed in the future.

In an extended version of this paper, we will include: (1) The analysis of the algorithm under the assumption of finite precision. (2) Variations of the condition number $\cond(\gamma)$, introduced in Definition~\ref{defi:conditiongammasampling}, and its analysis in terms of the bit-size and for a random $\gamma$. (3) A comparison with methods relying on approximate arc-length parametrizations or deterministic samples. 

\subsubsection*{Acknowledgements}
The last author is supported by a postdoctoral fellowship of the 2020 ``Interaction'' program of the Fondation Sciences Mathématiques de Paris. He is grateful to Evgenia Lagoda for moral support and Gato Suchen for the mathematical discussions regarding \S\ref{sec:inv}.

The last two authors are partially supported by ANR
JCJC GALOP (ANR-17-CE40-0009), the PGMO grant ALMA, and the PHC GRAPE.

The authors are thankful to Elias Tsigaridas and the referees for useful comments and suggestions that helped improved the quality of this paper.

\newpage

\begin{table}[!htb] 
\centering
\begin{tabular}{|c c||c|c|c||c|c|c|}\hline
 & & \multicolumn{3}{|c|}{$2^{-\ell} = 0.1$} & \multicolumn{3}{|c|}{$2^{-\ell} = 0.01$} \\ \hline
 $d$ & $n$ & $k$ & $Pr.\ T$ & $T/s$ & $k$ & $Pr.\ T$ & $T/s$ \\ \hline\hline
% &10 &28 &1.1381 &0.55718 &34 &1.4819 &0.63254 \\
%\hline
 &20 &20 &0.82 &0.53 &24 &0.94 &0.58 \\
\hline
%5 &30 &21 &1.0099 &0.58864 &26 &1.2728 &0.5708 \\
%\hline
 &40 &32 &1.08 &0.44 &39 &1.28 &0.47 \\
\hline
% &50 &29 &1.2392 &0.49543 &35 &1.8217 &0.6635 \\
%\hline
5 &60 &27 &0.92 &0.44 &32 &1.09 &0.48 \\
\hline
% &70 &21 &0.99189 &0.44784 &25 &0.9962 &0.52106 \\
%\hline
 &80 &29 &1.04 &0.43 &34 &1.18 &0.45 \\
\hline
%5 &90 &26 &1.0441 &0.43903 &31 &1.1578 &0.45479 \\
%\hline
 &100 &25 &1.24 &0.54 &29 &1.30 &0.59 \\
\hline\hline
%10 &10 &41 &1.7709 &0.65262 &49 &2.0279 &0.63895 \\
%\hline
 &20 &34 &1.34 &0.57 &40 &1.51 &0.60 \\
\hline
%10 &30 &41 &4.0049 &0.69586 &48 &3.2748 &0.65643 \\
%\hline
 &40 &21 &0.76 &0.40 &27 &0.89 &0.42 \\
\hline
%10 &50 &37 &1.7307 &0.60309 &43 &1.476 &0.56986 \\
%\hline
10 &60 &35 &1.17 &0.44 &41 &1.37 &0.47 \\
\hline
%10 &70 &44 &2.2928 &0.57902 &52 &2.1382 &0.61121 \\
%\hline
 &80 &38 &1.31 &0.53 &44 &1.65 &0.56 \\
\hline
%10 &90 &32 &1.23 &0.46499 &37 &1.3568 &0.48092 \\
%\hline
 &100 &37 &1.64 &0.64 &43 &2.03 &0.65 \\
\hline\hline
%15 &10 &53 &2.2728 &0.63769 &62 &2.5058 &0.66463 \\
%\hline
 &20 &35 &1.29 &0.59 &41 &1.70 &0.61 \\
\hline
%15 &30 &41 &1.6683 &0.56577 &47 &1.83 &0.59116 \\
%\hline
 &40 &46 &1.53 &0.48 &53 &1.70 &0.50 \\
\hline
%15 &50 &49 &1.7105 &0.53759 &56 &1.8555 &0.51107 \\
%\hline
15 &60 &47 &3.40 &0.49 &54 &3.54 &0.50 \\
\hline
%15 &70 &47 &1.8266 &0.56806 &54 &2.3101 &0.69196 \\
%\hline
 &80 &43 &1.65 &0.49 &50 &1.90 &0.53 \\
\hline
%15 &90 &48 &1.7766 &0.51149 &55 &2.1977 &0.57901 \\
%\hline
 &100 &49 &2.36 &0.63 &55 &2.36 &0.65 \\
\hline\hline
%20 &10 &53 &2.1689 &0.64377 &62 &2.5309 &0.65113 \\
%\hline
 &20 &67 &2.82 &0.70 &79 &4.41 &0.82 \\
\hline
%20 &30 &58 &1.8833 &0.50342 &67 &2.2659 &0.55506 \\
%\hline
 &40 &63 &2.07 &0.53 &72 &2.65 &0.57 \\
\hline
%20 &50 &50 &1.684 &0.55883 &57 &2.0777 &0.51765 \\
%\hline
20 &60 &51 &2.38 &0.66 &58 &2.24 &0.56 \\
\hline
%20 &70 &55 &2.3399 &0.55522 &63 &2.5354 &0.52426 \\
%\hline
 &80 &70 &2.60 &0.59 &80 &4.38 &0.76 \\
\hline
%20 &90 &57 &2.3356 &0.57021 &64 &2.8913 &0.65575 \\
%\hline
 &100 &56 &2.55 &0.62 &63 &3.24 &0.65 \\
\hline
\end{tabular}
\caption{\rm \label{tab:experiments}Random experiments: degree of the Chebyshev interpolant ($k$), preprocessing time ($Pr.\ T$) and time per sample after preprocessing ($T/s$) in terms of the degree ($d$), the ambient dimension ($n$) and the error ($2^\ell$)}
\end{table}

\begin{figure}[!htb]
	\centering
		\includegraphics[width=0.45\textwidth]{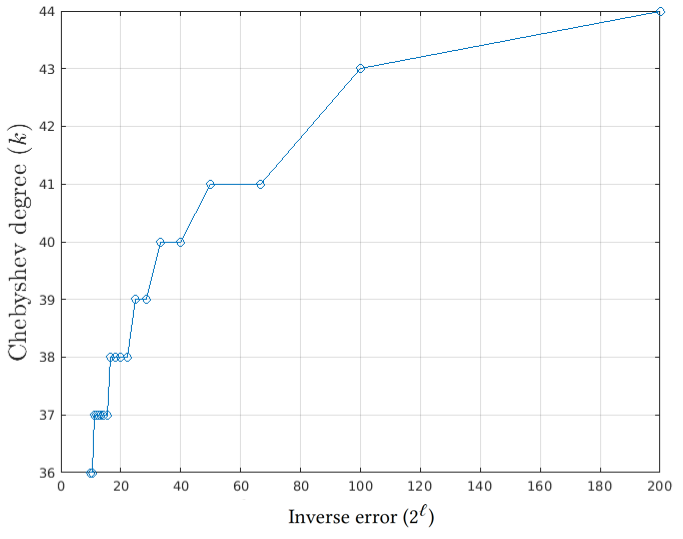}
	\caption{\rm Degree of the Chebyshev interpolant ($k$) against the inverse error ($2^\ell$) for a curve of degree 10 in $\bbR^{50}$.}\label{fig:degree_error}
\end{figure}

\begin{figure}[!htb]
	\centering
		\includegraphics[width=0.45\textwidth]{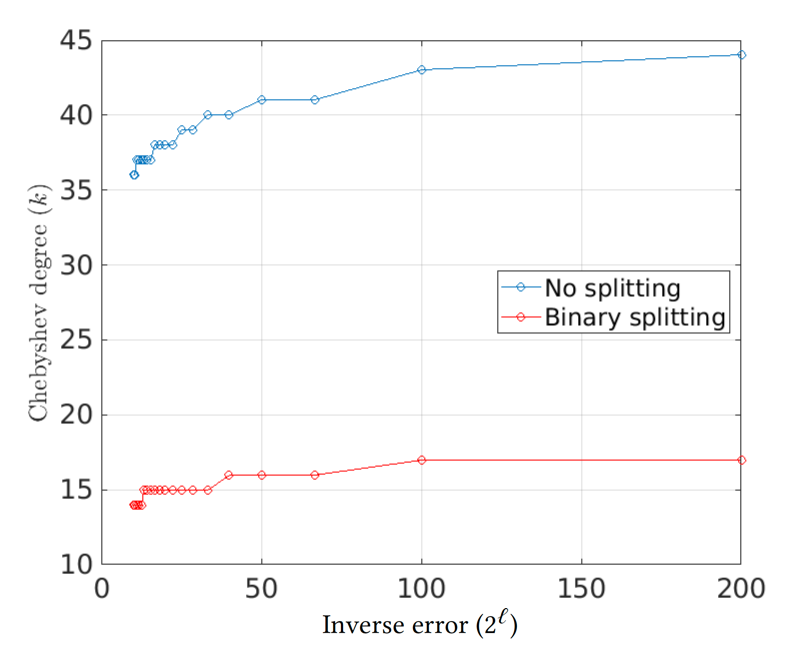}
	\caption{\rm Maximum degree of the Chebyshev interpolant ($k$) against the inverse error ($2^\ell$) for a curve of degree 10 in $\bbR^{50}$ without splitting the interval (blue) and splitting the interval in four equal intervals (red)}\label{fig:splitting}
\end{figure}

\begin{figure}[!htb]
	\centering
		\includegraphics[width=0.45\textwidth]{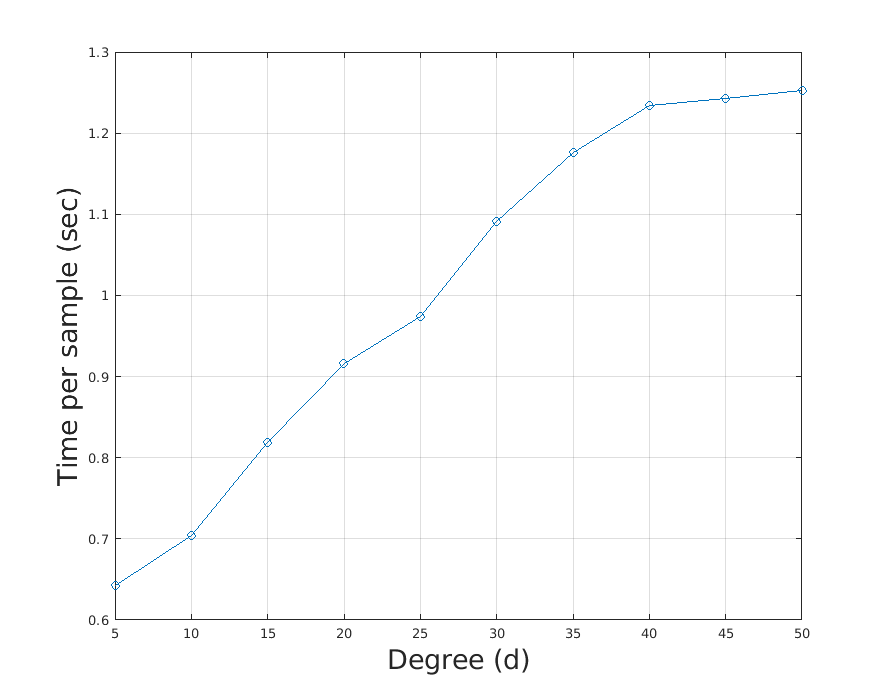}
	\caption{\rm Time per sample after preprocessing against the degree ($d$) for the 3D curve with each coordinate being the degree $d$ polynomial $1 + T + \cdots + T^d$.}\label{fig:degree_time}
\end{figure}

\begin{figure}[!htb]
	\centering
		\includegraphics[width=0.45\textwidth]{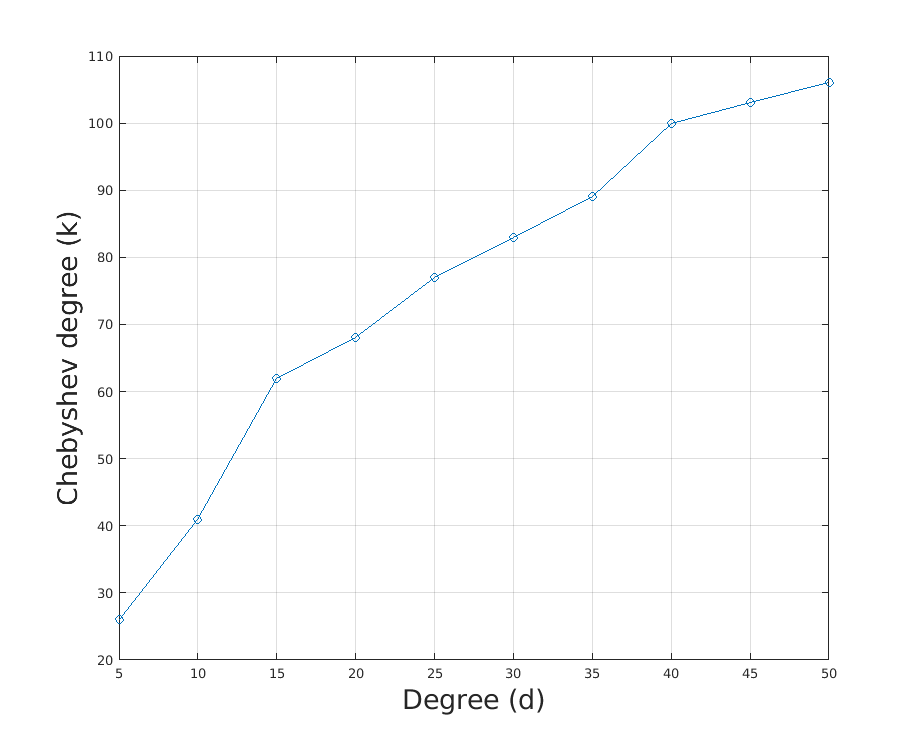}
	\caption{\rm Degree of the Chebyshev interpolant ($k$) against the degree ($d$) for the 3D curve with each coordinate being the degree $d$ polynomial $1 + T + \cdots + T^d$.}\label{fig:degree_degree}
\end{figure}

\bibliographystyle{ACM-Reference-Format}
\bibliography{biblio}

\newpage

%\begin{figure}[!htb]
%	\centering
%		\includegraphics[width=0.5\textwidth]{./figures/cheb_degree_error_comparison__3.png}
%	\caption{\rm Splitting vs. No Splitting}\label{fig:splitting}
%\end{figure}

\end{document}